\documentclass[11pt]{article}

\usepackage{epsfig, amssymb, amsmath}
\usepackage{latexsym, graphics, graphicx}
\usepackage{proof}
\usepackage{algorithmic}
\usepackage{algorithm}
\usepackage{amsthm} 
\usepackage{amsxtra}
\usepackage{float}
\usepackage{authblk}
\usepackage{mathtools}
\usepackage{mathptmx}
\usepackage{tikz}
\usepackage{enumitem}
\usetikzlibrary{positioning}
\tikzset{>=stealth}

\DeclareMathAlphabet{\mathcal}{OMS}{cmsy}{m}{n}

\setlist[itemize]{leftmargin=2.0cm}

\theoremstyle{plain}
\newtheorem{theorem}{Theorem}[section]
\newtheorem{Lemma}[theorem]{Lemma}
\newtheorem{corollary}[theorem]{Corollary}

\theoremstyle{definition}
\newtheorem{definition}[theorem]{Definition}

\newcommand{\lrstep}[2]{\mathrel{\xrightarrow{#1}_{#2}}}
\newcommand{\rlstep}[2]{\mathrel{\xleftarrow{#1}_{#2}}}
\newcommand{\ignore}[1]{}

\newcommand{\flr}{\rightarrow}

\title{Lynch-Morawska Systems on Strings}

\author{\small Daniel S. Hono II$^1$
and \small Paliath  Narendran$^1$ 
and \small Rafael Veras$^1$\\
\small
$^1$ University at Albany--SUNY (USA),  
            e-mail: {\tt \{dhono,pnarendran,rveras\}@albany.edu} \\
}

\ignore{
\title{Lynch-Morawska Systems on Strings}
\titlerunning{Notes on Lynch-Morawska Systems}

\author{
    Daniel S. Hono II\inst{1}
    \and
   	Paliath  Narendran\inst{1}
    \and
   	Rafael Veras\inst{1}    
}

\authorrunning{D. S. Hono II, P. Narendran, and R. Veras}
\institute{
    University at Albany---SUNY (USA)  \\
    \email{\{dhono,pnarendran,rveras\}@albany.edu}\\
}
}

\date{}

\begin{document}

\begin{titlepage}

{\vspace*{-1in}\hspace*{-.5in}
\parbox{7.25in}{
\setlength{\baselineskip}{13pt}
\makebox{\ }\hfill {\footnotesize College of Engineering and Applied Sciences} \\
\makebox{\ }\hfill {\footnotesize Computer Science Department} \\
\makebox{\ }\\
}

\vspace{-.775in}}

\epsfxsize=3.15in
\epsfclipon
\hspace{-0.5in}{\raggedright{
\epsffile{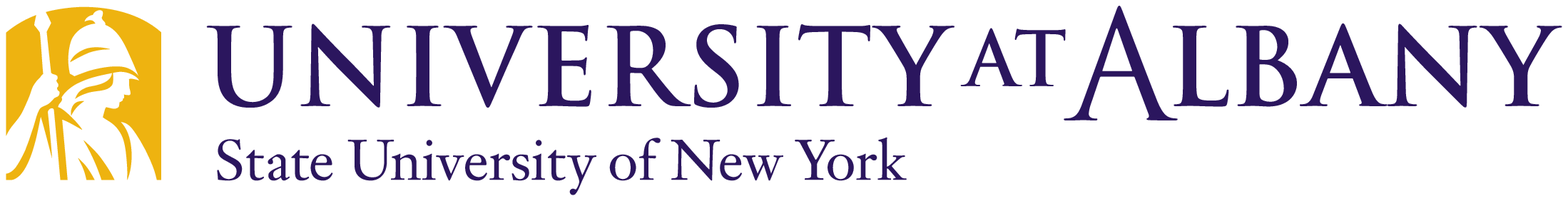}
}}

\vspace{2in}

\begin{center}
{\huge\bf Lynch-Morawska Systems on Strings}\\[+25pt]
\end{center}

\vspace{1.5in}

\begin{center}
{\large\bf 
Daniel S. Hono II\\[+3pt]
Paliath Narendran\\[+3pt]
Rafael Veras}\\[+3pt]

\end{center}
\end{titlepage}

\maketitle

\begin{abstract}
We investigate properties of convergent and forward-closed string
rewriting systems in the context of the syntactic criteria introduced
in~\cite{LynchMorawska} by Christopher Lynch and Barbara Morawska (we
call these $LM$-Systems). Since a string rewriting system can be
viewed as a term-rewriting system over a signature of purely monadic
function symbols, we adapt their definition to the string rewriting
case. We prove that the subterm-collapse problem for convergent and
forward-closed string rewriting systems is effectively
solvable. Therefore, there exists a decision procedure that verifies
if such a system is an $LM$-System. We use the same construction to
prove that the \emph{cap problem} from the field of cryptographic protocol
analysis, which is undecidable for general $LM$-systems, is decidable
when restricted to the string rewriting case.
\end{abstract}

\section{Introduction}
In this paper we investigate the properties of convergent and
forward-closed string rewriting systems. Our motivation comes from the
syntactic criteria defined by Christopher Lynch and Barbara Morawska
in~\cite{LynchMorawska}. They showed that for any term-rewriting
system $R$ that satisfies their criteria (which we call $LM$-Systems),
the unification problem modulo $R$ is solvable in polynomial
time. In~\cite{NotesOnBSM} it was shown that these conditions are
tight, i.e., relaxing any of them leads to NP-hard unification
problems.  It was also shown in~\cite{NotesOnBSM} that the
subterm-collapse problem for term-rewriting systems that satisfy all
of the other conditions of $LM$-Systems is undecidable.

In this current work, we show that the subterm-collapse problem is decidable
when restricted to convergent and forward-closed string rewriting systems. These
string rewriting systems can be viewed as term rewriting systems over a signature
of purely monadic function symbols. We give an analogous definition of $LM$-Systems
for string rewriting systems. Thus, given a forward-closed and convergent string rewriting system $T$ there is an algorithm that decides if $T$ is an $LM$-System. 

The construction used to show the decidability of the
subterm-collapse problem for forward-closed and convergent string
rewriting systems is also used to show that the \emph{cap problem}, an
important problem from the field of cryptographic protocol analysis~\cite{caps},
is also decidable for such string rewriting systems. This is in
contrast with some of our recent work that shows that the cap problem,
which is undecidable in general, remains undecidable when restricted
to general $LM$-Systems.

\ignore{All proofs in this extended abstract have been omitted for brevity. The
interested reader can consult the technical report for details\footnote{arxiv link}}

\section{Definitions}
We present here some notation and definitions. Only a few essential 
definitions are given here; for more details, the reader is 
referred to~\cite{BaaderNipkow} for term rewriting systems, and 
to~\cite{Botto} for string rewriting systems.

\medskip{} 

Let $\Sigma$ be a finite alphabet. As is usual, $\Sigma_{}^*$ stands
for the set of all strings over~$\Sigma$. The empty string is denoted
by~$\lambda$. For a string~$x$, $| x |$ denotes its length and $x^R$
denotes its reversal. A string $u$ overlaps with a string~$v$ iff there is a non-empty
\emph{proper} suffix of~$u$ which is a prefix of~$v$. For instance,
$aba$ overlaps with $acc$, but $aba$ does not overlap
with~$cca$. However, $aba$ overlaps with itself since $a$ is both a
prefix and a suffix of~$aba$. (See Fig~\ref{fig1}.)
\begin{figure}[h] 
\begin{center} 
\epsfig{file=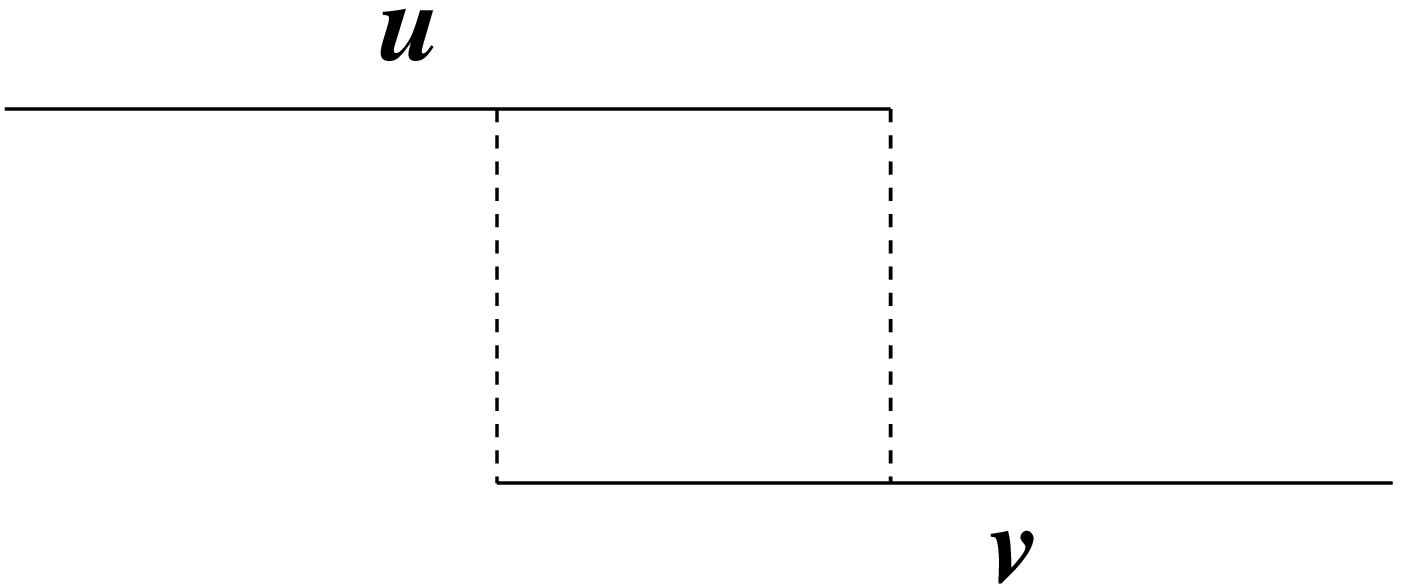, width=4.0in} 
\end{center}
\caption{overlap} \label{fig1}
\end{figure}

A string rewriting (rewrite) system 
(SRS) $R$ over this alphabet is a set of \emph{rewrite rules} of the form 
$l \to r$
where $l,r \in \Sigma_{}^*$; $l$ and $r$ are respectively called the left- 
and right-hand-side (\emph{lhs} and \emph{rhs}) of the rule. The
\emph{rewrite relation} on strings defined by the rewrite 
system $R$, denoted $\to_R$, is \[ \left\{ (xly , \, xry ) ~ \; | \; ~ x, y \in {\Sigma}_{}^*
\text{ and } (l, r) \in R \vphantom{b_b^a} \right\} \] The
reflexive and transitive closure of this relation is 
$\lrstep{*}{R}$.  An SRS $R$ is terminating iff there is no infinite chain 
of strings $s_i$, $i \in \mathbb{N}$, such that $s_i$ $R$-rewrites 
to $s_{i+1}$, that is to say $s_i \to_R s_{i+1}$.  
An SRS $R$ is \emph{confluent\/} iff for all strings $t$, $s_1$, $s_2$ such 
that $s_1 \rlstep{*}{R} t \lrstep{*}{R} s_2$  there exists a string $t'$ 
such that $s_1 \lrstep{*}{R} t' \rlstep{*}{R} s_2$.  An SRS $R$ 
is \emph{convergent\/} iff it is both terminating and confluent.

\medskip{} 
A string is irreducible with respect to~$R$ iff no rule of $R$ can be
applied to it.  The set of strings that are irreducible modulo $R$ is
denoted by $IRR(R)$. Note that this set is a regular language, since
$IRR(R) =  \Sigma_{}^* \smallsetminus 
  \{\Sigma_{}^* l_1 \Sigma_{}^* \,\cup ... \cup \, \Sigma_{}^* l_m
  \Sigma_{}^*\}$, 
where $l_1,\dots, l_m$ are the lhs of the rules in~$R$. 
A string $w'$ is an \emph{R-normal form\/} (or a
{\em normal form\/} if the rewrite system is obvious from the context)
of a string~$w$ for an SRS~$R$ if and only if $w \rightarrow_R^* w'$
and $w'$~is irreducible.  We write this
as $w \rightarrow_{R}^{!} w'$. An SRS $R$ is \emph{right-reduced} if
every right-hand side is in normal form.  An SRS~$T$ is said to be
{\em canonical\/} if and only if it is convergent and {\em
  inter-reduced,\/} i.e., it is right-reduced and,
besides, no lhs is a substring of another lhs.

\medskip{} 
Given a rewrite system $R$ and a set of strings $L$, $R_{}^*(L)$ 
is the set of all descendants of strings from~$L$, i.e., $
\{ x ~ | ~ \exists y \in L: \; y \rightarrow_{}^* x \}$, and
$R_{}^! (L)$ the set of normal forms of strings in $L$ for the rewrite 
system $R$.  Thus $R_{}^! (L)$ = $R_{}^* (L) \; \cap \; IRR(R)$.

String rewriting systems can be viewed as a restricted class of term 
rewriting systems where all functions are unary. As in~\cite{ANR}
a string $u$ over a given alphabet $\Sigma$ is viewed as a term over one variable 
derived from the {\em reversed\/} string of $u$; i.e., if $g, h \in \Sigma$, 
the string $gh$ corresponds to the term $h(g(x))$.  
(In other words, the unary operators defined by the symbols of a string 
are applied successively in the order in which these symbols appear in that 
string.) A string of the form~$wl$ where $w \in \Sigma_{}^*$ and
$l$ is a left-hand side is called a
\emph{redex}. 
A redex is \emph{innermost} if no proper prefix
of it is a redex.
The longest suffix of an innermost redex that is a left-hand
side in~$R$ is called its $l$-part and the remaining prefix is
referred to as its~$s$-part.

\medskip{} 
We will also need a special kind of normal form for strings, modulo 
any given SRS~$T$. With that purpose, we define, following 
S\'{e}nizergues~\cite{Seniz-PartialC}, a {\em leftmost-largest\/} reduction 
as follows: let~$\succ$ be a given total ordering on the alphabet~$\Sigma$ 
and ~$\succ_{L}^{}$ be its 
length~+~lexicographic extension\footnote{S\'enizergues  
refers to this as the {\em short-lex\/} ordering}.
A rewrite step $x l y \, \rightarrow \, x r y$ is {\em leftmost-largest\/}
if and only if 
(a)~$xl$ is an innermost redex, 
(b)~any other left-hand side
that is a suffix of~$xl$ is a suffix of~$l$ as well, (i.e., 
$l$~is the $l$-part of this redex) and
(c)~if $l \rightarrow r'$ is another
rule in the rewrite system, then $r' \succ_{L}^{} r$.
A string $w'$ is said to be a {\em leftmost-largest\/} ({\em ll-\/})
{\em normal form\/} of a string $w$ iff $w \rightarrow_{}^! w'$ using only 
leftmost-largest rewrite steps.
Given a terminating system~$T$, it holds that any string~$w$ 
has a \emph{unique} normal form produced by leftmost-largest 
rewrite 
steps alone, since every rewrite step is unique; 
this unique normal form will be denoted as $\rho_{T}^{} (w)$.

\ignore{
For a language~$L \subseteq \Sigma_{}^*$, $\widehat{\rho_{T}}(L)$
denotes the set of leftmost-largest normal forms of strings in~$L$,
i.e., \[ \widehat{\rho_T}(L) ~ = ~ \left\{ \, \rho_{T}^{} (w) ~ \left| ~
w \in L \, \vphantom{b^b} \right. \right\} \].  
}

Next, we define what it means for a string $x \in \Sigma^{+}$ 
to cause a subterm collapse. 

\noindent
\begin{definition}
Let $R$ be a convergent string rewriting system. A string~$x$ is said to
\emph{cause a subterm-collapse} if and only if there is a non-empty
string~$y$ such that $xy \rightarrow_R^* x$.
\end{definition}

\medskip{} 
Throughout the rest of the paper, $a, b, c, \dots, h$ will denote
elements of the alphabet $\Sigma$, and strings over $\Sigma$ will be
denoted as $l, r, u, v, w, x, y, z$, along with subscripts and 
superscripts.  

\ignore{
\medskip{} 
A  string rewrite system $T$ is said to be:
\begin{itemize}
\item[-] {\em special\/} iff the rhs of each rule in 
   $T$ is the empty string $\lambda$. \par 
\item[-] {\em monadic\/} iff the rhs of each rule in $T$ is either a single
 symbol or the empty string. \par 
\item[-] {\em dwindling\/} iff, for every rule $l \flr r$ in $T$, the rhs $r$ 
  is a {\em proper prefix\/} of its lhs~$l$. \par
\item[-] {\em optimally reducing\/}\footnote{For term
rewriting systems this notion was first introduced in~\cite{NPS}, and
has been generalized in~\cite{Comon-LundhD05}.} iff every rule in $T$ 
is optimally reducing,  i.e., iff the following holds:  
For every rule $ub \rightarrow v$ in $T$, with $u, v \in \Sigma_{}^*$,  
      $b \in \Sigma$, and for any string $z$, if $zv$ is reducible then 
      so is $zu$ (equivalently: if $zu$ is irreducible then so is $zv$). 
\item[-] \emph{forward-closed} iff every innermost redex can be reduced 
to its normal form \emph{in one step}.
\end{itemize}
}

\medskip{} 

A string rewrite system $T$ is said to be
\emph{forward-closed} iff every innermost redex can be reduced to its
normal form \emph{in one step}.

\ignore{
For instance, the string rewrite system $\{ ab \rightarrow ca \}$ is 
optimally reducing since, for any $x$, $xca$ is reducible if and only if 
$x$ itself is reducible. Similarly, the system $\{ ab \rightarrow ba , 
 ~ aa  \rightarrow a \}$
is optimally reducing too, though this requires a little more analysis.
On the other hand, $\{ ab \rightarrow b \}$ is not optimally
reducing since $aa$ is irreducible, but $ab$ is not.
}

We now give some preliminary results on convergent and forward-closed string rewriting systems. This first lemma shows that reducing all right-hand sides of rules in $R$ will preserve
the equivalence generated by $R$ as well as the properties that we are interested in. 

\begin{Lemma}
\label{RightReducedEquivalence}
Let $R$ be a convergent and forward-closed string
rewriting system, and let $l \rightarrow r$ be a rule in~$R$.
Then $\left(R \smallsetminus \left\{ l \rightarrow r \vphantom{l^b} \right\}\right) \cup
\left\{ l \rightarrow \rho_R^{}(r) \vphantom{l^b} \right\}$ is 
convergent, forward-closed and equivalent to~$R$.
\end{Lemma}

\begin{proof}
Let $R' = \left(R \smallsetminus \left\{ l \rightarrow r
\vphantom{l^b} \right\}\right) \cup \left\{ l \rightarrow \rho_R^{}(r)
\vphantom{l^b} \right\}$ where $R$ is convergent and forward-closed.
We make a few observations first. First of all,
since $R'$ contains the same left-hand sides as~$R$,
$IRR(R') = IRR(R)$. 
The set of redexes are the same too. Besides,
it is not hard to see that $\rightarrow_R^* \, \subseteq \, \rightarrow_{R'}^*$
since $l \rightarrow_R^{} r \rightarrow_R^{*} \rho_R^{}(r)$ for all rules~$l \rightarrow r$
in~$R$.

We first show that $R'$ and $R$ are equivalent. This is straightforward
since for every rule~$l \rightarrow r \in R$, $l$ and $r$ are
joinable modulo~$R'$ and vice versa. 
Thus~$\leftrightarrow_R^{*} \; = \; \leftrightarrow_{R'}^{*}$.

We next show that $R'$ is terminating given that $R$ is convergent.
For the sake of deriving a contradiction, assume that $R'$ is not
terminating. Then $ \exists t \in \Sigma^{*} : (t_i)_{i =
  0}^{\infty}$ and $t_i \rightarrow_{R'} t_{i+1}$ where $t_0 = t$.
Consider any $t_i \rightarrow_{R'} t_{i+1}$ step in the above
sequence. Then, by definition of reduction, there must be a rule $l
\rightarrow r \in R'$ such that: \[ t_i = xly \rightarrow xry = t_{i+1} \] Since 
no left-hand sides of rules in $R$ were altered in the
construction of $R'$, we can apply a corresponding rule in $R$. If the
rule $l \rightarrow \rho_{R}(r)$ was used, then we could replace the
above step with at most two reduction steps. Thus, we could construct
an infinite descending chain modulo~$R$, which is a contradiction.

Next, we show that $R'$ is confluent. Suppose it is not. Then
since $R'$ is terminating,
there must be a string~$t$ with two distinct normal forms
$t_1^{}$ and $t_2^{}$. But since $R$ is confluent and equivalent
to~$R'$, one of $t_1^{}$ and $t_2^{}$ must be \emph{reducible}
modulo~$R$. This is clearly a contradiction since
$IRR(R) = IRR(R')$.

Thus, $R'$ is convergent given that $R$ is convergent. 

It remains to show that $R'$ is forward-closed.
For this it is enough to
show that every innermost redex can be reduced to its normal form in a
\emph{single} reduction step. Let $x = x'l$ be an innermost redex modulo~$R'$
where $x, x' \in \Sigma^{*}$. Then $x$ is
also an innermost redex modulo~$R$. Since $R$ is forward-closed $x'r
\in IRR(R)$ for $l \rightarrow r \in R$. Thus, $x'r \in IRR(R')$ as
well, 
again since $IRR(R) = IRR(R')$.
\ignore{
Finally, to see that $R'$ is equivalent to $R$ to suffices to show
that $s \downarrow_{R} t \leftrightarrow s \downarrow_{R'} t$ where
$s, t \in \Sigma^{*}$ However, the same argument used above applies,
and we can simulate the reduction steps in both cases by simply
replacing the rules used modulo $R$ with the corresponding ones
modulo~$R'$.
}
\end{proof}

We next show that no left-hand sides of rules of a 
forward closed and convergent string-rewriting system
can be the same. 

\begin{corollary}
\label{DistinctLHS}
Let $R$ be a convergent, forward-closed and right-reduced string
rewriting system. Then no two distinct rules have the 
\emph{same} left-hand side.
\end{corollary}

\begin{proof}
 Suppose not. Let $l_i \rightarrow r_i \in R$ for $i \in \left\{1, 2\right\}$ 
 such that $l_1 = l_2$ but $r_1 \not = r_2$, but then $l \rightarrow r_1$ and $l \rightarrow r_2$
 as trivial reductions would not be joinable, as $r_1$ and $r_2$ are in normal form. 
\end{proof}

The next preliminary result shows that 
we can use leftmost-largest reduction 
steps to reduce an innermost redex to
its normal form in a single step. 

\begin{Lemma}
\label{LLSingleStep}
Let $R$ be a convergent, forward-closed and right-reduced string
rewriting system, and let $w$ be an innermost redex. Then
$w \rightarrow \rho_{R}^{} (w)$, i.e., $w$ reduces to its
normal form in \emph{one} leftmost-largest reduction step.
\end{Lemma}

\begin{proof}
  Let $w \in \Sigma^{*}$ be an innermost redex. Then $w = w'l$ for $w'
  \in \Sigma^{*}$ and by forward closure there must be some rule $l
  \rightarrow r \in R$ such that $l \rightarrow r$ reduces $w$ to its
  normal form in a single step. If this were not a leftmost-largest
  reduction, then there must be some other rule $l' \rightarrow r'
  \in R$ such that $w = w''l'$ is also an innermost-redex. By
  Corollary~\ref{DistinctLHS}, $l$ must be a proper suffix of $l'$ and
  $l'$ must be unique, then $w \rightarrow w''r' \in IRR(R)$ and $w
  \rightarrow w'r \in IRR(R)$, which contradicts the convergence
  of~$R$.
\end{proof}


\section{LM-Conditions for String Rewriting Systems}

We now give an equivalent definition of \emph{quasi-determinism} for
string rewriting systems $R$. This definition is adapted from that of~\cite{LynchMorawska}. 
We also define a \emph{right-hand side critical pair} for string-rewriting systems. Thus,
we are able to formulate the conditions of~\cite{LynchMorawska} in the context of string rewriting systems. 

\noindent
A string rewriting system $R$ is \emph{quasi-deterministic} if 
and only if
\begin{enumerate}
    \item No rule has $\lambda$ as its right-hand side

    \item No rule in $R$ is \emph{end-stable}---i.e., no rule has the
        same rightmost symbol on its left- and right-hand sides, and

    \item $R$ has no \emph{end pair repetitions}---i.e., no two rules in
        $R$ have the same unordered pair of rightmost symbols on their sides.
\end{enumerate}

\noindent
We define a \emph{right-hand-side critical pair} as follows: if
$l_1^{} \rightarrow r_1^{}$ and $l_2^{} \rightarrow r_2^{}$ are two
distinct rewrite rules and $r_2^{} = x r_1^{}$ for some~$x$ (i.e.,
$r_1^{}$ is a suffix of $r_2^{}$) then 
$\left\{ x l_1^{} , \, l_2^{} \right\}$ is a right-hand-side critical
pair. The set of all right-hand-side critical pairs is referred to as $RHS(R)$.

\medskip{}
It can be shown~\cite{NotesOnBSM} that

\begin{Lemma}
    Suppose $R$ is a convergent quasi-deterministic string rewriting system.
    Then $RHS(R)$ is not quasi-deterministic if and only
    if $RHS(R)$ has an end pair repetition.
    \label{lemma-quasi}
\end{Lemma}

\noindent
A string-rewriting system is \emph{deterministic} if and only if it is
non-subterm-collapsing and $RHS(R)$ is quasi-deterministic.\\

\noindent
A \emph{Lynch-Morawska string rewriting system} or \emph{LM-system}
is a convergent right-reduced string rewriting system~$R$ which satisfies
the following conditions:
\begin{itemize}
\item[(i)]   $R$ is non-subterm-collapsing, \\[-18pt]
\item[(ii)]  $R$ is forward-closed, and \\[-18pt]
\item[(iii)] $RHS(R)$ is quasi-deterministic. \\[-10pt]
\end{itemize}

\medskip{} 
In light of the results of~\cite{notesOnLmSystems}, a convergent 
string rewriting system~$R$ is an LM-system if and only if
$RHS(R)$ is quasi-deterministic and
\begin{itemize}
\item[(a)] $R$ is almost-left reduced (see~\cite{notesOnLmSystems}).
\item[(b)] There are no overlaps among the left-hand sides of~$R$.
\item[(c)] No lhs overlaps with a rhs.
\end{itemize}

We now work towards proving the main results of this paper. Namely, we will
show in the sequel below that the subterm-collapse problem for convergent and
forward-closed string rewriting systems is decidable.

The first of our results towards the above goal is below:

\begin{Lemma}
\label{NormalFormSeq}
Let $R$ be a convergent, forward-closed and
quasi-deterministic string rewriting system and $x, y, z \in IRR(R)$ such that
$xy \rightarrow_{}^! z$. Then there exist irreducible strings
$x = x_1^{}, x_2^{}, \ldots , x_{n}^{}, x_{n+1}^{}$,
$y_1^{}, y_2^{}, \ldots , y_n^{}, y_{n+1}^{}$ such that
\begin{enumerate}
\item $y = y_1^{} \ldots y_{n+1}^{}$,
\item $x_i^{} y_i^{}$ is an innermost redex for all $1 \le i \le n$,
\item $x_i^{} y_i^{} \rightarrow x_{i+1}^{} \;$ for all $1 \le i \le n$, $\; \;$ \emph{and}
\item $x_{n+1}^{} y_{n+1}^{} = z$.
\end{enumerate}
\end{Lemma}

\begin{proof}
The proof proceeds by induction on the number of rewrite steps along
the path from $xy$ to the normal form $z$. \\ \\
\noindent
\textbf{\underline{\underline{Basis.}}} Suppose $x, y, z \in IRR(R)$
such that $xy \rightarrow^{!} z$ in $k = 1$ steps.  That is, $xy
\rightarrow z$. Since $x, y \in IRR(R)$ there cannot be a redex that
is a substring of either $x$ or $y$ alone.  Hence there must be
strings $x', y' \in \Sigma^{*}$ and $l_1, l_2 \in \Sigma^+$ such
that \[ x = x'l_1, ~ y = l_2y' \] and $l_1l_2 = l$ for some $l
\rightarrow r \in R$.  Note that, since $R$ is convergent we may
assume that $x'l_1l_2$ is the shortest such redex.
\begin{figure}[h] 
\begin{center} 
\epsfig{file=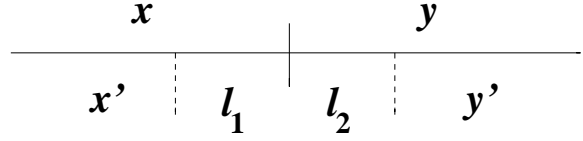, width=3.0in} 
\end{center}
\caption{first step} \label{fig2}
\end{figure}


$\therefore$ We can construct the following sequence: $x_1 = x'l_1, \;
x_2 = x'r, \; y_1 = l_2, \; y_2 = y'$ such that,
\begin{itemize}
	\item[(1)] $y = y_1y_2 = l_2y'$ 
	\item[(2)] $x_1y_1 = x'l_1l_2 = x'l$ is an innermost redex
	\item[(3)] $x_1y_1 \rightarrow x_2 = x'r$ 
	\item[(4)] $x_2y_2 = x'ry' = z$ 
\end{itemize}

\noindent
Above, $(1)$ and $(3)$ follow immediately from the definitions of
$x_1, y_1, x_2, y_2$. $(4)$ will follow after establishing
$(2)$. However, since $x'l_1l_2$ was chosen as the shortest such redex
appearing in $xy$ from the left and crossing the boundary between $x$
and~$y$, it must be an innermost redex. Thus, we have
established~$(2)$.

Now, since $R$ is forward-closed, $x'r$ can be assumed to be in normal
form.  Since $y'$ is a proper suffix of $y \in IRR(R)$, we get that
$y' \in IRR(R)$. Note that the above reductions are leftmost-largest.
Since every string has a unique leftmost-largest normal form modulo a
terminating string-rewriting system, and since $R$ is convergent, this
normal form must be $z$. \\

\textbf{\underline{\underline{Inductive Hypothesis.}}} Assume that the
result holds for all $x, y, z \in IRR(R)$ such that $xy
\rightarrow^{!} z$ in $k > 1$ steps. We show that it holds for strings
$x, y, z \in IRR(R)$ such that $xy \rightarrow^{!} z$ in $k+1$ steps.

Since $k > 1$, $\exists w \in \Sigma^+$ such that $xy \rightarrow w
\rightarrow^{!} z$. Note that $w \rightarrow^{!} z$ must take exactly
$k$ rewrite steps. As in the base case, since $x, y \in IRR(R)$ and
$xy$ is reducible, we have that $xy = x'l_1l_2y'$ where $x = x'l_1$,
$y = l_2y'$, and $l_1l_2 = l$ for some $l \rightarrow r \in R$ and
$x', y' \in \Sigma^*$. Since $R$ is convergent, we assume that $x'l$
is the leftmost prefix of $xy$ that is a redex.

We thus form the sequence: $x_1 = x'l_1, \; y_1 = l_2, \; x_2 = x'r,
\; y_2 = y'$. Since $x_1y_1$ is the leftmost redex of $xy$ it must be
the case that $x_1y_1$ is an innermost redex. Therefore, $x'r$ can be
assumed to be in normal form. Then $w = x'ry'$, and since $x'r \in
IRR(R)$ and $y \in IRR(R)$ we get that $w = uv$ for some $u, v \in
IRR(R)$. We can then apply the induction hypothesis to $u, v$, and $z$
to fill in the rest of the sequence with the desire properties.

Therefore we can conclude that the result holds for all $x, y, z \in IRR(R)$ such that $xy \rightarrow^{!} z$.  
\end{proof}

\ignore{
\noindent
\begin{definition}
Let $R$ be a convergent string rewriting system. A string~$x$ is said to
\emph{cause a subterm-collapse} if and only if there is a non-empty
string~$y$ such that $xy \rightarrow_R^* x$.
\end{definition}
}

An immediate consequence of the definition of subterm-collapse
given below. 

\medskip{}
\begin{Lemma}
\label{PrefixSubtermCollapse}
Let $R$ be a convergent forward-closed string rewriting system and $x,
y \in IRR(R)$ such that $xy \rightarrow_{}^! x$ and $y \neq \lambda$.
(Thus $x$ causes a subterm-collapse.) Let $y_1^{}$ be a prefix of~$y$.
Then $xy_1^{}$ causes a subterm-collapse.
\end{Lemma}

\begin{proof}
Let $x, y \in IRR(R)$ such that $xy \rightarrow^{!} x$. 
Let $\tilde{y}$ be any prefix of $y$. Thus $y = \tilde{y} y'$ for
some string~$y'$.

In order to generate a subterm-collapse with respect to $x\tilde{y}$, we must have 
a string $w \in \Sigma^+$ such that $x\tilde{y}w \rightarrow^{*} x\tilde{y}$. 
We construct such a string as follows: let $w = y' \tilde{y}$. 

Therefore, $x\tilde{y}w = x \tilde{y} y' \tilde{y} =
xy\tilde{y}$. Since $xy \rightarrow^{!} x$ 
we get $\, x\tilde{y}w = xy\tilde{y} \rightarrow^{*} x\tilde{y}$.
\end{proof}

We now prove that $R$ is subterm-collapsing
if and only if there is a right-hand side of
a rule in $R$ that causes a subterm collapse in 
the sense of the above definition. This lemma will
be key in showing the decidability of the subterm-collapse
problem as it allows us only to consider right-hand sides
of rules for possible sources of subterm-collapse. 

\medskip{}
\begin{Lemma}
\label{RHSsubtermcollapse}
Suppose $R$ is a convergent forward-closed string rewriting system. Then
$R$ is subterm-collapsing if and only if and only if there is a right-hand
side $r$ that causes a subterm-collapse.
\end{Lemma}

\begin{proof}
If there is a right-hand side that causes a subterm-collapse, then $R$ is subterm-collapsing. Towards 
  proving the ``only if'' direction, assume for the sake of deriving a contradiction that the result doesn't hold,
  i.e., $R$ is subterm-collapsing but no right-hand side causes a subterm-collapse.
  Then, let $w$ be one of the \emph{shortest} strings that causes a subterm-collapse.
  
  Since $w \not = \lambda$ it must be the case that $(\exists a \in \Sigma)(\exists w' \in \Sigma^{*}) : w = aw'$. Also, since
  $w$ is assumed to cause a subterm-collapse, $(\exists z \in \Sigma^{+}) : wz = aw'z \rightarrow^{*} w = aw'$.
  There are thus two cases to consider: 
  either $a$ is involved in the reduction, i.e., $a$ is in the $l$-part of a redex, or it is not. 
 
  Suppose $a$ is involved in the reduction. By Lemma~\ref{RightReducedEquivalence}, without
  loss of generality we can assume that $R$ is right-reduced.
  Since $a$ is the first letter of $w$ and $a$ is involved in some reduction step, there must be a
  prefix~$z'$ and a corresponding suffix $z''$ of~$z$ such that $wz' = aw'z' \rightarrow^{*} aw'' \rightarrow r$
  and $rz'' \rightarrow^{*} w$ for some~$w''$. That is, $aw''$ is a redex as well as its 
  $l$-part, i.e., $aw'' = l$ for some $l \rightarrow r \in R$.
  But by the previous lemma (Lemma~\ref{PrefixSubtermCollapse}), $wz'$ and hence $r$ causes
  a subterm-collapse. This contradicts our assumption.
  
  Now, suppose $a$ is not involved in the reduction sequence. Then it must be
  that $w'z \rightarrow^{*} w'$. Thus,
  $w'$ causes a subterm-collapse and $|w'| < |w|$, which contradicts the minimality of $w$. 
\end{proof}


The main lemma of this section appears below. It 
gives us that a certain language, parameterized by 
two strings $u, v \in \Sigma^{*}$, is a deterministic
context-free language. We prove this by constructing
a deterministic pushdown automaton to recognize this 
language.

\medskip{}
\begin{Lemma}
\label{DPDA}

Let $R$ be a convergent, right-reduced, and forward-closed string rewriting system, $u, v \in IRR(R)$, 
and $\# \not \in \Sigma$. Then the language \[ \mathcal{L}_{u, v}^{} \; = \; \left\{ \vphantom{b^b}
\, w \# ~ \; \mid \; ~ uw \rightarrow_{}^! v, \; ~ w \not = \lambda \, \right\} \] is
a deterministic context-free language over $(\Sigma \cup \{\#\})^*$

\end{Lemma}

\begin{proof}
 We design a deterministic pushdown automaton (DPDA) 
 $\mathcal{M}$ that recognizes $\mathcal{L}_{u, v}^{}$. In the sequel, let
 $x$ denote the contents of $\mathcal{M}'s$ stack from bottom to top.
 
 Initially, $\mathcal{M}$ pushes a special symbol, $\$$, onto the
 stack (which serves as a bottom marker) and then pushes~$u$. Thus,
 the contents of the stack after the initialization steps are~$\$u$.
 
 Then, we design a transition system based on two cases. Either
 pushing the symbol $a \in \Sigma$ completes a redex or it does
 not. That is,
 \begin{itemize}
  \item[1] $(x, \; a) \mapsto (xa, \; \lambda)$ if $xa$ is not a redex, or
  \item[2] $(x, \; a) \mapsto (x'r_0, \; \lambda)$ if $xa = x'l_0$ 
where $x'$ is the $s$-part and $l_0$ the $l$-part of~$xa$ (i.e., $l_0$ the \emph{longest} left-hand side 
in~$R$ that is a suffix of~$xa$).
 \end{itemize}

 $\mathcal{M}$ will carry out the above transitions by pushing symbols
 of $w$ (which is initially on the tape) and reducing each redex that
 appears. When $\mathcal{M}$ reaches the $\#$ symbol, by
 Lemma~\ref{NormalFormSeq} if $uw \rightarrow^{!} v$ then the contents
 of the stack must be~$\$v$. This can be checked by~$\mathcal{M}$.
 
 Finally, $\mathcal{M}$ can be created by building an Aho-Corasick
 automaton, $\mathcal{K}$, for the set
$\left\{ \vphantom{b^b} l_1, l_2, \ldots, l_n \right\}$
as given, for instance, in~\cite{Crochemore-Rytter}.
Then $\mathcal{M}$ can simulate
 $\mathcal{K}$ on its stack by essentially restarting $\mathcal{K}$
 whenever $\mathcal{K}$ accepts.
\end{proof}

As a consequence of the above Lemma~\ref{DPDA} the subterm collapse problem is decidable
for convergent, foward-closed, string-rewriting systems. 

\begin{corollary}
\label{subtermdec}
 The following decision problem: \\ 
 
 \noindent

\underline{\underline{Given: }} A convergent, forward-closed, right-reduced SRS $R$.

\underline{\underline{Question: }} Is $R$ subterm-collapsing? \\

is effectively solvable.
\end{corollary}

\begin{proof}
 A decision procedure can be constructed by creating, for each $l
\rightarrow r \in R$, a DPDA~$\mathcal{M}_{r}$ such that
$L(\mathcal{M}_{r}) = \mathcal{L}_{r, r}$ by lemma~\ref{DPDA}. 
$\mathcal{M}_r$ can then be converted into an equivalent context-free grammar~$G_{r}$. Then
$L(G_r) = \varnothing$ if and only if $r$ does not cause a
subterm-collapse. By Lemma~\ref{RHSsubtermcollapse} this is enough to
conclude that $R$ is not subterm-collapsing in general. Finally,
deciding if a CFG generates the empty language is decidable,
therefore, the overall problem is decidable as well.
\end{proof}

Note also, that the construction outlined above can be carried out
\emph{in polynomial time.} Thus, not only is the above
subterm-collapse problem for convergent, forward-closed string
rewriting systems decidable, it is efficiently decidable. This is in
contrast to the results of~\cite{NotesOnBSM} where it was shown that
checking if a given term-rewriting system is subterm-collapsing, even
when the system satisfies all of the other Lynch-Morawska conditions, is
undecidable.

We can therefore conclude that the problem of verifying if a convergent
and forward-closed string rewriting system (or a term rewriting system
over a signature of monadic function symbols) is an $LM$-system is decidable.

As another corollary of the above result, we get that the cap problem 
for convergent, forward-closed, string-rewriting systems is also decidable. 
This problem, also known as the deduction problem, is often studied in the field
of symbolic cryptographic protocol analysis. 

\begin{corollary}
\label{cap}

The Cap Problem:
\indent
 \begin{itemize}
\item[\underline{\underline{Given:}}] A convergent, forward-closed 
 string-rewriting system $R$, a string $u \in \Sigma^{+}$ 
 (representing the intruder knowledge) and a secret $v \in \Sigma^{+}$.
  
\item[\underline{\underline{Question:}}] Does there exists a string $w \in
 \Sigma^{+}$ (called a cap term) such that $uw \rightarrow_{R}^{!} v$? 
 \end{itemize}
is decidable. 
 
\end{corollary}

\begin{proof}
 The construction is essentially the same as that in the proof of Corollary~\ref{subtermdec}. This time
 a $DPDA$ is constructed, using Lemma~\ref{DPDA}, for the language~$\mathcal{L}_{u, v}$. 
\end{proof}

The result of Corollary~\ref{cap} is contrasted with the fact that, 
for general term-rewriting systems, the cap problem is known to be undecidable. 
The cap problem remains undecidable even when restricted to $LM$-Systems. 
The above results shows, in the monadic case, if $R$ is convergent and 
forward-closed, then the problem becomes decidable.

\newpage
\bibliographystyle{plain}
\bibliography{ref2}

\end{document}